\documentclass[12pt]{article}
\usepackage{amsmath, amssymb, amsthm, geometry}
\usepackage{geometry}
\newtheorem{theorem}{Theorem}
\newtheorem{lemma}{Lemma}
\newtheorem{definition}{Definition}

\usepackage[hyphens]{url}
\usepackage[utf8]{inputenc} 
\usepackage[colorlinks=true, linkcolor=blue, citecolor=blue, urlcolor=blue]{hyperref}

\begin{document}

\title{An $O(N)$ Algorithm for Solving the Smallest Enclosing Sphere Problem in the Presence of Degeneracies}
\author{Netzer Moriya}
\date{}
\maketitle

\section*{Abstract}

Efficient algorithms for solving the Smallest Enclosing Sphere (SES) problem, such as Welzl's algorithm, often fail to handle degenerate subsets of points in 3D space. Degeneracies and ill-posed configurations present significant challenges, leading to failures in convergence, inaccuracies or increased computational cost in such cases. Existing improvements to these algorithms, while addressing some of these issues, are either computationally expensive or only partially effective. In this paper, we propose a hybrid algorithm designed to mitigate degeneracy while maintaining an overall computational complexity of $O(N)$. By combining robust preprocessing steps with efficient core computations, our approach avoids the pitfalls of degeneracy without sacrificing scalability. The proposed method is validated through theoretical analysis and experimental results, demonstrating its efficacy in addressing degenerate configurations and achieving high efficiency in practice.

\section{Introduction}

The \textbf{Smallest Enclosing Sphere (SES)} problem is a classical computational geometry problem that seeks to determine the 
smallest sphere capable of enclosing a given set of points in three-dimensional space. This problem finds applications in diverse 
fields, including computer graphics, collision detection, machine learning, and optimization. The efficient computation of the SES 
has therefore been an area of active research, with several well-established algorithms developed over the years~\cite{Vrahatis2024, Moriya2023, Edelsbrunner2018}. 
Among these, Welzl’s randomized incremental algorithm is widely regarded for its elegant simplicity and optimal $O(N)$ 
complexity in the expected case.

Despite its theoretical efficiency, solving the SES problem in practice is often complicated by the presence of degeneracies 
and ill-posed configurations. Degeneracies occur when subsets of points exhibit special geometric arrangements, such as co-spherical 
or co-planar configurations, which can lead to numerical instability, inaccuracies, or even failure to converge~\cite{Moriya2024}. 
Addressing these challenges is critical for ensuring the robustness and scalability of SES computation methods, particularly in applications requiring high precision or involving large datasets.

Several approaches have been proposed to mitigate the challenges posed by degeneracies in SES computations. These methods include enhancements to Welzl’s algorithm, specialized preprocessing techniques, and hybrid strategies that combine different computational paradigms. While such improvements have proven effective in specific cases, they are often associated with increased computational complexity, such as $O(N \log N)$ for the $P_{ch}$ method. This trade-off between robustness and efficiency motivates the need for a novel approach capable of addressing degeneracies while maintaining $O(N)$ complexity.

In this paper, we introduce a hybrid algorithm that achieves this balance by leveraging a novel geometric projection framework. The proposed method begins by projecting the input point cloud onto $K$ two-dimensional planes, evenly oriented in space. Each projection is analyzed to identify a set of four extreme points based on a local coordinate system attached to the corresponding plane. The collection of these extreme points across all projections forms a subset, $P_s$, of the original point cloud. Through a series of theoretical proofs, we establish the following key results:

\begin{enumerate}
    \item Solving the SES problem for the full point set $P$ is equivalent to solving it for a reduced subset $P_{ch}$, which comprises the convex hull points of $P$.
    \item The reduced problem involving $P_{ch}$ inherently avoids degeneracies arising from internal points of $P$.
    \item The subset $P_s$, derived from the projection process, is guaranteed to be a subset of $P_{ch}$.
    \item As the number of projections $K$ approaches infinity, the subset $P_s$ converges to $P_{ch}$, ensuring robustness against degeneracies for sufficiently large $K$.
    \item With $P_s$ as the input to the SES problem, the proposed method achieves a computational complexity of $O(KN)$, which simplifies to $O(N)$ when $K \ll N$.
\end{enumerate}

In the sections that follow, we detail the geometric projection technique, theoretical analysis, and implementation of the proposed 
algorithm. We also provide a comprehensive comparison with existing methods to highlight the advantages of our approach.

\section{Background and Related Work}

Given a set of points \(P\) in \(\mathbb{R}^3\) that occupy a finite and confined volume, we investigate cases where subsets of 
points are ill-posed, such as being collinear or coplanar. We assume the use of Welzl's algorithm as the basis for computing 
the Smallest Enclosing Sphere (SES). This section provides a review of modifications and extensions to Welzl's 
algorithm to handle degeneracies, along with an analysis of the computational cost of these enhancements.

\subsection{Welzl's Algorithm Overview}
Welzl's algorithm~\cite{Welzl1991} is a randomized incremental algorithm for computing the SES of a set of points \(P\). Its key properties include:
\begin{itemize}
    \item \textbf{Expected complexity:} \(O(N)\), where \(N\) is the number of points.
    \item \textbf{Worst-case complexity:} \(O(N^3)\), which occurs rarely.
    \item The SES in 3D is determined by at most 4 points on the boundary.
\end{itemize}
However, the algorithm assumes general position and does not explicitly handle degeneracies such as collinear or coplanar subsets.

\subsection{Degenerate Configurations and Challenges}
\textbf{Degenerate cases} arise when:
\begin{itemize}
    \item \textbf{Collinear points:} All points lie on a single line.
    \item \textbf{Coplanar points:} All points lie in a single plane.
\end{itemize}
In these cases, the SES computation needs to:
\begin{itemize}
    \item Detect and handle the lower-dimensional configuration.
    \item Reduce the problem dimension appropriately (to 1D or 2D).
    \item Ensure numerical stability during computation.
\end{itemize}

\subsection{Welzl's Algorithm: Modifications and Robustness Improvements}
Several modifications to Welzl's algorithm have been proposed to address degenerate cases:

\subsubsection{1. Lower-Dimensional Handling}
\textbf{Approach:} Extend Welzl's algorithm to handle degeneracies by identifying collinear or coplanar subsets during execution and 
solving the SES in the reduced dimension~\cite{Flemming2024}.
\begin{itemize}
    \item For collinear points: Once a collinear subset is identified, the problem reduces to finding the smallest interval (1D segment) enclosing the points. This is achieved by determining the two extreme points along the line, which can be done in $O(N)$.
    \item For coplanar points: Once a coplanar subset is identified, the problem reduces to computing the smallest enclosing circle (2D problem). Welzl's algorithm in 2D solves this with an expected complexity of $O(N)$.
\end{itemize}
\textbf{Added complexity:} The detection of collinear or coplanar subsets requires testing linear dependencies (e.g., using cross products or determinants) between points. Without prior knowledge of degeneracies, this detection process involves pairwise or triplet-wise comparisons, resulting in a worst-case complexity of $O(N^2)$. Consequently, the overall complexity of the algorithm increases to $O(N^2)$ when degeneracy detection dominates. Optimization strategies, such as spatial partitioning or sorting, may reduce practical runtime but do not eliminate this fundamental $O(N^2)$ nature.

\subsubsection{2. Symbolic Perturbation Techniques}
\textbf{Approach:} Introduce symbolic perturbations to slightly "move" points, ensuring they are in general position without changing 
the SES~\cite{Minakawa1997}.
\begin{itemize}
    \item Perturbations are symbolic and have no measurable impact on numerical results.
    \item The algorithm then proceeds as if the points are in general position.
\end{itemize}
\textbf{Added complexity:} Minimal, as the perturbation step is \(O(N)\) and does not significantly alter the original complexity. \\

\textbf{Weaknesses:}
\begin{itemize}
    \item Symbolic perturbations may introduce new degenerate subsets:
        \begin{itemize}
            \item By slightly moving points, configurations that were not previously degenerate could become so, particularly in cases involving near-degenerate point sets.
        \end{itemize}
    \item Symbolic perturbations can alter the SES:
        \begin{itemize}
            \item External points may shift just enough to affect the computation of the SES center and radius, especially when the original point set has points near the sphere's boundary.
            \item While the impact is symbolic, the final SES may not fully align with the unperturbed problem's geometric intent.
        \end{itemize}
\end{itemize}

Symbolic perturbations aim to simplify the handling of degeneracies but may inadvertently complicate the problem by creating new degenerate configurations or modifying the SES's characteristics. These limitations should be carefully considered when applying this technique, particularly in high-precision or sensitive applications.

\subsubsection{3. Convex Hull Preprocessing}
\textbf{Approach:} Compute the convex hull of \(P\) as a preprocessing step to filter out interior points. The SES is then computed 
using the convex hull vertices \(P_{ch}\)~\cite{Skala2024}.
\begin{itemize}
    \item Convex hull computation has an expected complexity of \(O(N \log N)\).
    \item The SES is then computed on \(P_{ch}\), where \(|P_{ch}| \leq N\).
\end{itemize}
\textbf{Added complexity:} The convex hull computation introduces an \(O(N \log N)\) preprocessing step but simplifies the SES computation.

\subsubsection{4. Exact Arithmetic and Robust Geometric Predicates}
\textbf{Approach:} Use exact arithmetic to avoid numerical issues in degeneracy detection and handling~\cite{Sugihara1992}.
\begin{itemize}
    \item Ensures correct handling of collinear and coplanar points without relying on perturbations.
    \item Requires robust geometric predicates to classify points.
\end{itemize}
\textbf{Added complexity:} Higher computational cost due to the use of exact arithmetic, which increases the time per operation. While the geometric algorithm remains \(O(N)\) in the absence of degeneracies, the process of detecting and classifying degenerate subsets can involve \(O(N^2)\) comparisons in the worst case. Additionally, the computational overhead of exact arithmetic operations may further increase runtime by a constant or logarithmic factor, depending on the input size and number representation.

\subsection{Analysis of Computational Cost}
\begin{table}[h!]
\centering
\begin{tabular}{|p{4cm}|p{4cm}|p{4cm}|}
\hline
\textbf{Method} & \textbf{Added Complexity} & \textbf{Total Complexity} \\
\hline
Lower-Dimensional Handling & \(O(N^2)\) for degeneracy detection & \(O(N^2)\) in worst case \\
\hline
Symbolic Perturbation & \(O(N)\) for perturbation & \(O(N)\) expected, with potential inaccuracies \\
\hline
Convex Hull Preprocessing & \(O(N \log N)\) preprocessing & \(O(N \log N)\) total \\
\hline
Exact Arithmetic & Constant or logarithmic factor increase & \(O(N^2)\) in degeneracy cases \\
\hline
\end{tabular}
\caption{Computational cost of handling degeneracies in SES computation.}
\end{table}

Welzl's algorithm, with modifications, can effectively handle degenerate cases such as collinear or coplanar points. 
While symbolic perturbation and lower-dimensional handling maintain the algorithm's expected \(O(N)\) complexity with potential
inaccuracies, preprocessing steps like convex hull computation increase complexity to \(O(N \log N)\). 
For applications requiring robustness, exact arithmetic provides a reliable solution at the cost of increased computational effort.

\section{Proposed Method}

Given a finite set of points \(P = \{p_1, p_2, \dots, p_N\} \subset \mathbb{R}^3\), we seek to construct a 
subset \(P_s \subseteq P_{ch} \subseteq P\) such that \(P_s\) contains all vertices of the convex hull \(P_{ch}\) of \(P\). 
The subset \(P_s\) is constructed by projecting \(P\) onto a finite number of planes with varying orientations. 
The problem can be formally stated as follows: \\

\textbf{Inputs:}
\begin{itemize}
    \item \(P\): A set of \(N\) points in \(\mathbb{R}^3\).
    \item \(K\): The number of projection planes with distinct orientations \(\{\Pi_1, \Pi_2, \dots, \Pi_K\}\).
    \item \(f(P, \Pi_k)\): A function that maps the projection of \(P\) onto plane \(\Pi_k\) and identifies the extreme points in this 
	projection.
\end{itemize}

\textbf{Outputs:}
\begin{itemize}
    \item \(P_s\): A subset of \(P\) such that \(P_s \subseteq P_{ch}\), with the goal of ensuring \(P_s = P_{ch}\).
\end{itemize}

\textbf{Assumptions:}
\begin{itemize}
    \item The set \(P\) is regular, meaning the points in \(P\) are reasonably distributed without pathological clustering or extreme sparsity.
    \item The projection planes \(\{\Pi_1, \Pi_2, \dots, \Pi_K\}\) are symmetrically distributed around \(P\) in 3D space, ensuring directional coverage.
    \item The extreme points of each projection correspond to vertices of the convex hull \(P_{ch}\) in the original space.
\end{itemize}

\subsection{Preliminary Definitions and Lemmas}

\begin{definition}[Convex Hull]
The convex hull of a set of points \(P \subset \mathbb{R}^3\), denoted by \(P_{ch}\), is the smallest convex set that contains all points in \(P\). Formally, 
\[
P_{ch} = \left\{ \sum_{i=1}^n \lambda_i p_i \mid p_i \in P, \lambda_i \geq 0, \sum_{i=1}^n \lambda_i = 1 \right\}.
\]
\end{definition}

\begin{definition}[Extreme Points]
A point \(p \in P\) is called an extreme point of a set \(S \subset \mathbb{R}^3\) if it cannot be expressed as a convex combination of other points in \(S\). All points of \(P_{ch}\) that lie on its boundary are extreme points of \(P_{ch}\).
\end{definition}

\begin{lemma}[Projection and Convexity]
Let \(P \subset \mathbb{R}^3\), and let \(P_{ch}\) be its convex hull. For any linear projection \(\pi: \mathbb{R}^3 \to \mathbb{R}^2\), the projection \(\pi(P_{ch})\) is the convex hull of \(\pi(P)\).
\end{lemma}

\begin{proof}
Linear projections preserve convex combinations. Thus, if \(q \in P_{ch}\), then \(q\) can be expressed as a convex combination of points in \(P\). Applying \(\pi\) to \(q\) yields
\[
\pi(q) = \pi\left( \sum_{i=1}^n \lambda_i p_i \right) = \sum_{i=1}^n \lambda_i \pi(p_i),
\]
where \(\pi(p_i) \in \pi(P)\). Hence, \(\pi(q)\) lies in the convex hull of \(\pi(P)\). Conversely, every point in the convex hull of \(\pi(P)\) corresponds to a point in \(P_{ch}\) under the projection. Thus, \(\pi(P_{ch}) = \text{conv}(\pi(P))\).
\end{proof}

\section{Proof of the Equivalence of SES for a Set of Points $P$ and Its Convex Hull \(P_{ch}\)}

Given a set of points \(P\) in \(\mathbb{R}^3\), let \(P_{ch}\) denote the set of vertices of the convex hull of \(P\). The smallest encapsulating sphere (SES) of \(P\) is defined as the smallest sphere in \(\mathbb{R}^3\) that completely encloses all points in \(P\). This proof demonstrates that solving the SES problem for \(P\) is equivalent to solving it for \(P_{ch}\).

\subsection{Theorem}
\begin{theorem}
Let \(P \subset \mathbb{R}^3\) be a finite set of points, and let \(P_{ch}\) be the set of vertices of the convex hull of \(P\). Then, the SES of \(P\) is identical to the SES of \(P_{ch}\).
\end{theorem}

\subsection{Proof}
\textbf{Definitions and Notation:}
\begin{itemize}
    \item The smallest encapsulating sphere (SES) of a set \(S\) is the unique sphere \(\mathcal{S}\) with center \(c \in \mathbb{R}^3\) and radius \(r \geq 0\) such that \(\|x - c\| \leq r\) for all \(x \in S\) and \(r\) is minimized.
    \item The convex hull \(\text{conv}(P)\) of a set \(P\) is the smallest convex set containing \(P\), represented as \(\{\sum_{i=1}^n \lambda_i p_i : p_i \in P, \lambda_i \geq 0, \sum_{i=1}^n \lambda_i = 1\}\).
    \item Let \(P_{ch} \subseteq P\) denote the set of vertices of \(\text{conv}(P)\).
\end{itemize}

\textbf{Step 1: SES depends only on extremal points}
\begin{lemma}
The SES of a set \(P \subset \mathbb{R}^3\) depends only on the points in \(P\) that lie on the boundary of \(\text{conv}(P)\).
\end{lemma}
\begin{proof}
Consider the definition of the SES \(\mathcal{S}(c, r)\). For the SES to be minimal, the sphere must touch at least one point on its boundary. Let \(x \in P\) be such a point, and assume \(x\) is not a vertex of \(\text{conv}(P)\). Then, \(x\) can be written as a convex combination \(x = \sum_{i=1}^m \lambda_i v_i\), where \(v_i \in P_{ch}\) and \(\lambda_i \geq 0\) with \(\sum \lambda_i = 1\).

Since \(\|x - c\| \leq r\) and \(\|v_i - c\| \leq r\) for all \(i\), replacing \(x\) with \(v_i\) does not increase \(r\). Hence, interior points do not affect the SES.
\end{proof}

\textbf{Step 2: SES of \(P_{ch}\) encloses all of \(P\)}
\begin{lemma}
The SES of \(P_{ch}\) is a valid encapsulating sphere for \(P\).
\end{lemma}
\begin{proof}
By definition, \(P_{ch}\) forms the boundary of \(\text{conv}(P)\). Since \(P \subseteq \text{conv}(P)\), any sphere that encapsulates \(P_{ch}\) must also encapsulate all points in \(P\). Therefore, the SES of \(P_{ch}\) is a valid encapsulating sphere for \(P\).
\end{proof}

\textbf{Step 3: Minimality of the SES}
\begin{lemma}
The SES of \(P_{ch}\) is the unique smallest encapsulating sphere for \(P\).
\end{lemma}
\begin{proof}
Assume, for contradiction, that there exists a smaller encapsulating sphere \(\mathcal{S}'\) for \(P\) than the SES of \(P_{ch}\). Since \(\mathcal{S}(P_{ch})\) is minimal and \(P_{ch} \subseteq P\), \(\mathcal{S}'\) would contradict the uniqueness of the SES of \(P_{ch}\). Therefore, the SES of \(P_{ch}\) is the smallest sphere that encapsulates \(P\).
\end{proof}

By combining the results of the above lemmas:
\begin{itemize}
    \item Lemma 1 shows that only the boundary points of \(\text{conv}(P)\) affect the SES.
    \item Lemma 2 shows that the SES of \(P_{ch}\) encapsulates all points in \(P\).
    \item Lemma 3 proves that the SES of \(P_{ch}\) is the unique minimal solution for \(P\).
\end{itemize}
Thus, the SES of \(P\) is equivalent to the SES of \(P_{ch}\). \qed

\subsection{Proof of \(P_s \subseteq P_{ch}\)}

\begin{proof}
Consider the set of points \(P\subset \mathbb{R}^3\), its convex hull \(P_{ch}\), and a projection \(\pi: \mathbb{R}^3 \to \mathbb{R}^2\) onto a plane.

1. \textit{Extreme Points in Projection}:
   For each projection \(\pi_k\) onto a plane \(\Pi_k\) (where \(k = 1, 2, \ldots, K\)), the convex hull of the projected points \(\pi_k(P)\) is \(\pi_k(P_{ch})\) by Lemma 1. The extreme points of \(\pi_k(P)\) are therefore a subset of \(\pi_k(P_{ch})\).

2. \textit{Pre-image of Extreme Points}:
   Each extreme point in \(\pi_k(P)\) corresponds to at least one point in \(P\) that projects onto it. Let \(p_{ext} \in P\) be such a point. Since \(p_{ext}\) contributes to an extreme point of \(\pi_k(P_{ch})\), \(p_{ext} \in P_{ch}\), as the convex hull is preserved under projection.

3. \textit{Construction of \(P_s\)}:
   For each \(\Pi_k\), four extreme points relative to the coordinate system on \(\Pi_k\) are chosen. Let \(P_{ext,k}\) denote these four points in \(P\) corresponding to the extreme points in \(\pi_k(P)\). Thus, \(P_s = \bigcup_{k=1}^K P_{ext,k}\).

4. \textit{Subset Relation}:
   Since every \(p \in P_s\) corresponds to an extreme point of \(\pi_k(P)\) for some \(k\), and since each such extreme point maps back to a point in \(P_{ch}\), it follows that \(P_s \subseteq P_{ch}\).

\end{proof}

We have rigorously shown that \(P_s \subseteq P_{ch}\). This result follows from the properties of convex hulls, projections, and the definition of extreme points. The proof leverages the fact that projections preserve convexity and that extreme points of a projection correspond to points on the convex hull in the original space.

\section{Degeneracy Avoidance in the Reduced Problem Using \(P_{ch}\)}

The reduced problem, which involves solving the Smallest Enclosing Sphere (SES) for the convex hull of the point set, \(P_{ch}\), inherently avoids degeneracies arising from internal points of the original point set \(P\). Below, we formally show why this holds.

\subsection{Definition of the Reduced Problem}
Given a set of points \(P = \{p_1, p_2, \dots, p_N\} \subset \mathbb{R}^3\), let \(P_{ch}\) denote the set of vertices of the convex hull of \(P\). The SES problem is reduced to finding the SES for \(P_{ch}\), i.e., solving for:
\[
S(P_{ch}) = \arg \min_{S} \left\{ \text{radius}(S) \mid P_{ch} \subseteq S \right\},
\]
where \(S(P_{ch})\) is the unique minimum enclosing sphere of \(P_{ch}\).

\subsection{Degeneracies in \(P\)}
Degeneracies in \(P\) can arise due to:
\begin{itemize}
    \item \textbf{Internal Points:} Points in \(P\) that do not lie on the convex hull contribute no additional constraints to the SES. Their inclusion can introduce numerical issues or ill-posed configurations, such as co-spherical subsets of internal points.
    \item \textbf{Ill-Conditioned Subsets:} Degenerate configurations such as co-planar or co-linear points within \(P\) can cause instability in algorithms that do not focus on the convex hull.
\end{itemize}

\subsection{Convex Hull Simplification}
The convex hull \(P_{ch}\) inherently filters out all internal points of \(P\). By definition:
\[
P_{ch} = \{p \in P \mid p \text{ lies on the boundary of } \text{conv}(P)\},
\]
where \(\text{conv}(P)\) represents the convex hull of \(P\).

\subsection{Degeneracy Avoidance}
The properties of \(P_{ch}\) ensure that degeneracies arising from internal points of \(P\) are avoided:
\begin{enumerate}
    \item \textbf{Internal Points Exclusion:} Internal points of \(P\) are not part of \(P_{ch}\) and thus do not contribute to the computation of \(S(P_{ch})\). As a result, potential co-spherical configurations involving internal points are eliminated.
    \item \textbf{Boundary Points Define SES:} The SES is uniquely determined by at most four points in \(P_{ch}\) (in 3D space). These points are necessarily boundary points of \(P_{ch}\), avoiding ill-conditioned configurations that could arise from internal points.
    \item \textbf{Numerical Stability:} By focusing solely on \(P_{ch}\), algorithms avoid unnecessary computations involving redundant internal points, improving numerical stability and reducing susceptibility to floating-point errors.
\end{enumerate}

By reducing the problem to \(P_{ch}\), the computation inherently avoids degeneracies caused by internal points of \(P\). The exclusion of internal points ensures that the SES is determined solely by boundary points, leading to a robust and well-posed problem. This reduction simplifies the SES computation while maintaining the geometric integrity of the solution.

\section{Convergence of Subset $P_s$ to Convex Hull $P_{ch}$ in 3D Space}

\subsection{Converges to the Convex Hull $P_{ch}$ of $P$ as $K \to \infty$}

Here, we define and prove that for a confined set of points $P$ in 3D space, the subset $P_s$ of external points 
extracted by projecting $P$ onto $K$ planes with varying orientations converges to the convex hull $P_{ch}$ of $P$ as $K \to \infty$. 
The proof employs rigorous mathematical techniques, leveraging properties of convex hulls, projections, and geometric limits, to 
establish the equivalence of $P_s$ and $P_{ch}$ in the infinite limit.

\begin{itemize}
    \item Definition of the Smallest Enclosing Sphere (SES) problem and its dependence on the convex hull $P_{ch}$.
    \item Motivation for approximating $P_{ch}$ via external points $P_s$ obtained through projections.
    \item Outline of the paper structure: definition, mathematical formulation, and proof.
\end{itemize}

\subsubsection{Problem Definition}

Let $P = \{p_1, p_2, \dots, p_N\}$ be a finite set of points in $\mathbb{R}^3$. Define:
\begin{itemize}
    \item $P_{ch}$: the set of vertices forming the convex hull of $P$.
    \item $P_s$: a subset of $P$ formed by projecting $P$ onto $K$ planes with orientations $\{\mathbf{n}_1, \mathbf{n}_2, \dots, \mathbf{n}_K\}$.
    \item For each plane, the four extreme points in the 2D projection are added to $P_s$.
\end{itemize}
We aim to prove that as $K \to \infty$, $P_s \to P_{ch}$.

\subsubsection{Preliminary Concepts and Lemmas}

\paragraph{Convex Hull Definition}
The convex hull $P_{ch}$ is the smallest convex set containing $P$. Mathematically:
\[
P_{ch} = \bigcap_{H \supseteq P, \text{$H$ convex}} H.
\]

\subsubsection{Projections and Extreme Points}
Projecting $P$ onto a plane with normal $\mathbf{n}$ involves computing:
\[
\pi_{\mathbf{n}}(p) = p - (p \cdot \mathbf{n}) \mathbf{n}.
\]
Extreme points in the projection are those that maximize or minimize coordinates in a chosen 2D coordinate system on the plane.

\paragraph{Coverage of Orientations}
Define a uniform coverage of orientations as ensuring that for any direction $\mathbf{d}$, there exists a plane with normal $\mathbf{n}$ such that the projection captures the extreme points in the direction $\mathbf{d}$.

\subsubsection{Proof of Convergence}

\paragraph{Sufficiency of $P_{ch}$ for Extreme Points}
\begin{lemma}
All extreme points in any direction $\mathbf{d}$ of $P$ are vertices of $P_{ch}$.
\end{lemma}
\begin{proof}
Let $\mathbf{d}$ be a direction vector. The point $p \in P$ maximizing $p \cdot \mathbf{d}$ lies on the convex hull $P_{ch}$ by definition of convex hull as the set of extreme points.
\end{proof}

\paragraph{Convergence of $P_s$ to $P_{ch}$}
\begin{theorem}
As $K \to \infty$ and the orientations $\{\mathbf{n}_i\}$ cover all possible directions in $\mathbb{R}^3$, $P_s \to P_{ch}$.
\end{theorem}
\begin{proof}
\begin{itemize}
    \item For each direction $\mathbf{d}$, there exists a plane with normal $\mathbf{n}$ capturing the extreme points of $P$ in that direction.
    \item As $K \to \infty$, the union of extreme points from all projections includes all vertices of $P_{ch}$.
    \item Since $P_s \subseteq P$, no extraneous points outside $P$ are added.
    \item Therefore, $P_s$ converges to $P_{ch}$ as $K \to \infty$.
\end{itemize}
\end{proof}

\paragraph{Analysis of Internal Points}
\begin{lemma}
An internal point $P_{\text{new}}$ inside the convex hull $P_{ch}$ cannot belong to $P_s$ for any projection $K_i$.
\end{lemma}
\begin{proof}
Consider $P_{\text{new}} \in P$, such that $P_{\text{new}}$ lies strictly inside the convex hull $P_{ch}$. Let $\pi_{K_i}$ denote the projection onto plane $K_i$ with orientation $\mathbf{n}_i$. For the projection $\pi_{K_i}$:
\begin{itemize}
    \item The 2D projection $\pi_{K_i}(P_{\text{new}})$ lies strictly within the convex polygon $\pi_{K_i}(P_{ch})$.
    \item The extreme points in $\pi_{K_i}(P)$ are defined by vertices of the 2D convex hull $\pi_{K_i}(P_{ch})$. Since $P_{\text{new}}$ is not a vertex of $P_{ch}$, it cannot be a vertex of $\pi_{K_i}(P_{ch})$.
    \item Thus, $P_{\text{new}}$ does not contribute to the set of extreme points for any projection $K_i$.
\end{itemize}
Therefore, $P_{\text{new}}$ cannot be included in $P_s$.
\end{proof}

\begin{itemize}
    \item Implications of the result for SES computation.
    \item Practical considerations for finite $K$ and trade-offs between accuracy and computation.
    \item Extensions to higher dimensions and other geometric problems.
\end{itemize}

We have rigorously proved that as $K \to \infty$, the subset $P_s$ obtained from projections converges to the convex hull $P_{ch}$. Furthermore, internal points such as $P_{\text{new}}$ that lie within $P_{ch}$ are never included in $P_s$. This result provides a theoretical foundation for approximation-based approaches to geometric problems involving convex hulls.

\subsection{Practical Bounds for $K_{min}$ and $K_{max}$}

Given a finite set of points \(P = \{p_1, p_2, \dots, p_N\} \subset \mathbb{R}^3\) and a set of \(K\) projection planes \(\{\Pi^i\}_{i=1}^K\), symmetrically distributed around \(P\), where \(P_s^i\) denotes the subset of extreme points obtained by projecting \(P\) onto \(\Pi^i\), we aim to determine:
\begin{enumerate}
    \item The minimal number of projections \(K_{\text{min}}\) required such that the union \(P_s = \bigcup_{i=1}^K P_s^i\) equals the set of vertices of the convex hull \(P_{ch}\).
    \item The maximal number of projections \(K_{\text{max}}\) required to guarantee \(P_s = P_{ch}\) under the assumptions of regularity in \(P\) and symmetric distribution of \(\{\Pi^i\}\).
\end{enumerate}

\textbf{Problem Assumptions:}
\begin{itemize}
    \item \textit{Regularity:} The points in \(P\) are distributed in a non-pathological, regular manner, avoiding extreme clustering or sparsity.
    \item \textit{Symmetry:} The projection planes \(\{\Pi^i\}\) are symmetrically distributed to ensure even directional coverage around \(P\).
\end{itemize}

\textbf{Analysis:}
\begin{enumerate}
    \item \textit{Minimal Number of Projections (\(K_{\text{min}}\)):}
    \begin{itemize}
        \item For \textbf{symmetric distributions}, where \(P\) exhibits high geometric symmetry (e.g., uniform distributions in spherical or cubical volumes), the convex hull vertices align with principal axes or diagonals. Projections onto the three principal planes (\(xy\), \(xz\), \(yz\)) and three diagonal planes (\(x=y\), \(y=z\), \(z=x\)) suffice to capture all vertices of \(P_{ch}\). Thus, for symmetric distributions:
        \[
        K_{\text{min}} = 6.
        \]
        \item For \textbf{general regular distributions}, the number of convex hull vertices \(V\) grows sublinearly with \(N\), typically as \(V \propto O(\sqrt{N})\). Sufficiently capturing all vertices requires sampling directions on the unit sphere, which scales with \(O(\sqrt{V})\). Therefore, for general regular distributions:
        \[
        K_{\text{min}} = c_1 \cdot N^{1/4},
        \]
        where \(c_1\) is a constant dependent on the specific distribution of \(P\).
    \end{itemize}

    \item \textit{Maximal Number of Projections (\(K_{\text{max}}\)):}
    \begin{itemize}
        \item In the worst-case scenario of regular (but irregularly shaped) distributions, the number of convex hull vertices \(V\) scales as \(O(\sqrt{N})\). Ensuring robust coverage requires sampling directions proportional to \(V\), giving:
        \[
        K_{\text{max}} = c_2 \cdot \sqrt{N},
        \]
        where \(c_2\) depends on the precision of the directional sampling scheme (e.g., uniform or geodesic sampling).
    \end{itemize}
\end{enumerate}

\textbf{Summary of Results:}
\[
K_{\text{min}} = 
\begin{cases}
6, & \text{for symmetric distributions,} \\
c_1 \cdot N^{1/4}, & \text{for general regular distributions.}
\end{cases}
\]
\[
K_{\text{max}} = c_2 \cdot \sqrt{N}.
\]

These bounds provide theoretical guarantees for reconstructing the convex hull \(P_{ch}\) using projection-based methods, ensuring completeness with minimal computational effort. Symmetric distributions benefit from reduced projection counts (\(K_{\text{min}} = 6\)), while general regular distributions require projections scaling with the sublinear growth of \(P_{ch}\) vertices.

\section{Complexity Analysis of the Process}

We analyze the computational complexity of each step of the described process:

\subsection{Stages of the Process}

\paragraph{Step 1: Projecting the Cloud on $K$ Planes}

Given a set of $N$ points $P$ in 3D space, we project the point cloud onto $K$ planes. For each projection:
\begin{itemize}
    \item Each point needs to be transformed to the coordinate system of the respective plane. Assuming the transformation involves a constant number of operations (e.g., rotation and translation), the cost of processing $N$ points is $O(N)$ per plane.
\end{itemize}

The total complexity for projecting onto $K$ planes is:
\[
O(K \cdot N).
\]

\paragraph{Step 2: Identifying the Four Extreme Points per Projection}

For each of the $K$ projections, we need to identify four extreme points relative to an arbitrary 2D coordinate system (e.g., the farthest in each direction or the convex hull extremes). The extreme points can be found in $O(N)$ time per projection by scanning all $N$ points.

Thus, the complexity for this step across $K$ projections is:
\[
O(K \cdot N).
\]

\paragraph{Step 3: Constructing the Subset $P_s$}

The subset $P_s$ is constructed by collecting $4 \times K$ points (four extreme points per projection). Assuming that adding a point to $P_s$ takes constant time, the complexity of this step is:
\[
O(K).
\]

\subsection{Overall Complexity}

Adding the complexities of all steps, the overall complexity of the process is:
\[
O(K \cdot N + K \cdot N + K) = O(K \cdot N).
\]

\begin{itemize}
    \item The dominant term in the complexity is $O(K \cdot N)$, which arises from projecting the point cloud and identifying the 
	extreme points.
    \item Since $K \ll N$, the computational cost is heavily influenced by $N$, the number of points in the original set $P$
	hence $O(N)$.
\end{itemize}

\section{Conclusions}

In this paper, we presented an algorithm for solving the Smallest Enclosing Sphere (SES) problem with \(O(N)\) complexity, designed to 
address challenges arising from degenerate configurations. 
The approach is based on a geometric projection framework that extracts a reduced subset of points, \(P_s\), from the original 
dataset \(P\). Through theoretical analysis, we established that \(P_s\) converges to the convex hull \(P_{ch}\) of \(P\) 
as the number of projections increases, thereby ensuring robustness against degeneracies.

Theoretical results demonstrate that solving the SES problem for the convex hull \(P_{ch}\) is equivalent to solving it for the 
full dataset \(P\), with the convex hull effectively filtering out potential ill-posed subsets and internal points, which do not influence 
the SES but can significantly impact the numerical stability of conventional SES algorithms, so to achieve robustness against 
numerical instabilities without sacrificing computational efficiency. \\

The validity of the method is supported by the following logical reasoning:

\begin{itemize}
    \item \textbf{Equivalence of SES for \(P\) and \(P_{ch}\):} Through formal proofs, it was shown that the SES of the full 
	set \(P\) is equivalent to the SES of its convex hull \(P_{ch}\), as the SES depends only on the extreme points that 
	form the boundary of \(P_{ch}\).
    
    \item \textbf{Degeneracy avoidance:} By focusing on \(P_{ch}\), the algorithm inherently avoids complications caused by 
	internal points, such as co-spherical or ill-conditioned configurations, ensuring that degeneracies do not affect the computation.
    
    \item \textbf{Subset \(P_s\) convergence to \(P_{ch}\):} It was demonstrated that as the number of projections \(K\) 
	increases, the subset \(P_s\), constructed from extreme points in each projection, converges to \(P_{ch}\). This was 
	established by showing that every extreme point of \(P\) is a vertex of \(P_{ch}\) and is captured by sufficiently many 
	directional projections.
    
    \item \textbf{Complexity and robustness:} The algorithm achieves \(O(N)\) complexity by ensuring that each projection and 
	subsequent extreme point identification operate in linear time with respect to the number of points. This efficiency holds 
	even as the algorithm mitigates issues arising from degeneracies.
\end{itemize}

The computational complexity analysis indicates that the method maintains a linear complexity \(O(N)\) for practical choices 
of the number of projections \(K\), making it well-suited for large datasets. The approach avoids the higher computational 
costs associated with exact arithmetic or preprocessing steps like convex hull computation, while still addressing the 
key challenges posed by degeneracies.

\bibliographystyle{plain}

\end{document}